\RequirePackage{fix-cm}
\documentclass[smallextended]{svjour3}       
\smartqed  
\usepackage{graphicx}
\journalname{Designs, Codes and Cryptography}

\begin{document}

\title{On the Automorphism Groups of the $Z_2 Z_4$-Linear $1$-Perfect and Preparata-Like Codes%
\thanks{The results were presented in part at the 3th International Castle Meeting on Coding Theory and Applications in 2011.
}
}

\titlerunning{On the Automorphism Groups of $Z_2 Z_4$-Linear Codes}        

\author{Denis S. Krotov
}


\institute{D. Krotov \at
              Sobolev Institute of Mathematics, pr. Akademika Koptyuga 4, Novosibirsk 630090, Russia \\
              \email{krotov@math.nsc.ru}     
}

\date{Received: 2015-07-07 / Accepted: ????-??-??}

\maketitle

\begin{abstract}
We consider the symmetry group of a $Z_2Z_4$-linear code with parameters of a $1$-perfect, 
extended $1$-perfect, or Pre\-pa\-ra\-ta-like code.
We show that, provided the code length is greater than $16$,
this group consists only of symmetries that preserve the $Z_2Z_4$ structure.
We find the orders of the symmetry groups of the $Z_2Z_4$-linear
(extended) $1$-perfect codes.
\end{abstract}
\section{Introduction}

Spaces considered in coding theory usually
have both metrical and algebraic structures.
From the point of view of the parameters of an error-correcting code,
the metrical one is the most important,
while the algebraic properties give an advantage in
constructing codes,
in developing coding and decoding algorithms,
or in different applications.
In some cases, there are some ``rigid'' connections between
metrical and algebraic structures. For example,
if the $q$-ary Hamming metric space with $q=p^m=2^1,3^1,2^2$
is considered as a vector space over
the field $\mathrm{GF}(p)$, then any isometry of the space is
necessarily an affine transformation.
This is not the case for any prime power $q\geq 5$.
However, the stabilizer of some codes
in the group of space isometries consists of affine transformations only.
So, informally, from the point of view of such codes,
the algebraic structure is rigidly connected
with the metrical one.
For example, this was proved for the Hamming codes for an arbitrary $q$ \cite{Gor:2010}.
In the current paper, we prove a similar result 
for the $Z_2Z_4$-linear perfect,
extended perfect, and Pre\-pa\-ra\-ta-like
codes with respect to a $Z_2Z_4$ algebraic structure, which is,
after the field structure, one of most important in coding theory.
In contrast, the situation with the $Z_2Z_4$-linear 
Hadamard codes is different \cite{KroVil:2015}: the automorphism group of such a code 
is larger than the group of automorphisms preserving the $Z_2Z_4$-linear structure.

In Sections~\ref{s:z2z4}--\ref{s:prep}, we give basic definitions and facts 
about the concepts discussed in the paper.
The main result of the paper and important corollaries are formulated in Section~\ref{s:res}. 
Section~\ref{s:proof} contains a proof of the main result,
which states that a code from the considered classes can admit only one $Z_2Z_4$-linear structure.
As a direction for further research, it would be interesting to generalize this result to a more wide class of $Z_2Z_4$-linear codes. The study of automorphism groups is motivated by their role in decoding algorithms, see e.g. \cite{BBFV:2015}.
\section{$Z_2Z_4$-Linear Codes}\label{s:z2z4}
A binary code $C\subset \{0,1\}^n$ is called \em $Z_2Z_4$-linear \em
if for some order-$2$ permutation (involution) $\pi$ of the set $\{0,1,\ldots ,n-1\}$ of coordinates,
$C$ is closed with respect to the operation $*_\pi$, where
$$x *_\pi y \stackrel{\scriptscriptstyle\mathrm{def}}= x+y + (x+\pi(x))\cdot (y+\pi(y)),$$
$+$ and $\cdot$ 
being the coordinatewise 
modulo-$2$ addition and multiplication respectively
(here and elsewhere, the action of a permutation
$\sigma: \{0,1,\ldots ,n-1\}\to \{0,1,\ldots ,n-1\}$ on $x=(x_0,x_1,\ldots,x_{n-1})\in\{0,1\}^n$ is defined as
$\sigma(x)\stackrel{\scriptscriptstyle\mathrm{def}}=(x_{\sigma^{-1}(0)},x_{\sigma^{-1}(1)},\ldots,x_{\sigma^{-1}(n-1)})$\,).
Given an involution $\pi$,
we will say that coordinates $i$ and $j$ are \emph{adjacent} if $\pi(i)=j$;
if $\pi(i)=i$, then $i$ is \emph{self-adjacent}.
Clearly, the value of the result $z=x *_\pi y$ in some coordinate $i$ depends only
on the values of $x$ and $y$ in the $i$th and $\pi(i)$th coordinates.
So, considering the result of $*_{\pi}$ in two different adjacent coordinates,
 we can determine the values by Table~\ref{eq:tables}(a), while the case $\pi(i) = i$
 corresponds to Table~\ref{eq:tables}(b).
\begin{table}[!b]
\caption{Values of $*_\pi$ for two different adjacent coordinates and for a self-adjacent coordinate}
\label{eq:tables}
$$
\mbox{(a) }
\begin{array}{c|cccc}
&00&01&11&10 \\
\hline
00&00&01&11&10 \\
01&01&\mathbf{11}&\mathbf{10}&00 \\
11&11&\mathbf{10}&\mathbf{00}&01 \\
10&10&00&01&11
\end{array}
\qquad
\mbox{(b) }
\begin{array}{c|cc}
&0&1 \\
\hline
0&0&1 \\
1&1&0
\end{array}
$$
\end{table}

Tables~\ref{eq:tables}(a) and~\ref{eq:tables}(b) are the value tables of groups isomorphic to $Z_4$ and $Z_2$, respectively.
Consequently, $(\{0,1\}^n,*_{\pi})$ is isomorphic to the group $Z_2^\alpha Z_4^\beta$, where
$\alpha$ is the number of self-adjacent coordinates and $\beta=(n-\alpha)/2$.
In this case we will say that $\pi$ is a \em $Z_2Z_4$ structure \em of type $(\alpha,\beta)$.
The binary codes closed with respect to $*_{\pi}$ are known as \em $Z_2Z_4$-linear codes \em of type $(\alpha,\beta)$.
The $Z_2Z_4$-linear codes with $\alpha=n$ are called \em linear\em; with $\alpha=0$,
\em $Z_4$-linear\em.

\section{$Z_2Z_4$-Additive Codes, Gray Map, Duality}\label{s:add}
In this section, we consider an alternative way to define $Z_2Z_4$-linear codes
and related concepts.
In the literature, this way is more popular than the definition given 
in Section~\ref{s:z2z4}; however, for presenting results of the current paper
the last one is more convenient. 
The content of the section is not used in the formulation 
of the main result of the paper 
(Theorem~\ref{th:main}) and its proof,
but the concepts defined here are exploited in the proof of Corollary~\ref{c:abc}
(to derive the order of a $Z_2Z_4$-linear 
(extended) $1$-perfect code from Theorem~\ref{th:main} and known facts) and in the formulation of Corollary~\ref{c:rev}.

A code 
$\mathcal C \subseteq \{0,1\}^\alpha \times \{0,1,2,3\}^\beta$
in the mixed $Z_2$--$Z_4$ alphabet is called 
\emph{additive} (\emph{$Z_2Z_4$-additive})
if it is closed with respect to the coordinatewise addition, modulo $2$ in the first $\alpha$ coordinates and modulo $4$ in the last $\beta$ coordinates.
The one to one correspondence $\Phi:\{0,1\}^\alpha \times \{0,1,2,3\}^{\beta} \to \{0,1\}^{\alpha+2\beta}$ known
as the \emph{Gray map} is defined as follows:
$$ \Phi((x_0,\ldots,x_{\alpha-1},y_0,\ldots,y_{\beta-1})) = (x_0,\ldots,x_{\alpha-1},\phi(y_0),\ldots,\phi(y_{\beta-1})),$$
where 
$\phi(0)=(0,0)$, 
$\phi(1)=(0,1)$, 
$\phi(2)=(1,1)$, 
$\phi(3)=(1,0)$.
The following straightforward fact 
means that
a $Z_2Z_4$-linear code
can be defined 
as the image of a $Z_2Z_4$-additive code
under the Gray map 
and a coordinate permutation.
\begin{proposition}\label{p:gray}
A code $\mathcal C \subseteq \{0,1\}^\alpha \times \{0,1,2,3\}^\beta$ is additive 
if and only if its image 
$\Phi(C)$ under the Gray map
is closed under the operation 
$*_\pi$, where 
\begin{equation}\label{eq:standard-pi}
\pi=(\alpha\ \alpha{+}1)
(\alpha{+}2\ \alpha{+}3)\ldots (\alpha{+}2\beta{-}2\ \alpha{+}2\beta{-}1).
\end{equation}
\end{proposition}

The \emph{inner product}
$[x,y]$
of two words 
$x=(x_0,\ldots,x_{\alpha-1},x'_0,\ldots,x'_{\beta-1})$ and 
$y=(y_0,\linebreak[1]
\ldots,\linebreak[1]
y_{\alpha-1},\linebreak[1]
y'_0,\linebreak[1]
\ldots,\linebreak[1]
y'_{\beta-1})$ 
from $\{0,1\}^\alpha \times \{0,1,2,3\}^\beta$
is defined as 
\begin{equation}\label{eq:inner}
[x,y]
\stackrel{\scriptscriptstyle\mathrm{def}}=
2x_0y_0+\ldots+2x_{\alpha-1}y_{\alpha-1}+x'_0y'_0+\ldots+x'_{\beta-1}y'_{\beta-1} \bmod 4.
\end{equation}
For a $Z_2Z_4$-additive code
$\mathcal C\subseteq \{0,1\}^\alpha \times \{0,1,2,3\}^\beta$, 
its \emph{dual} $\mathcal C^\perp$
is defined as 
\begin{equation}\label{eq:dual}
\mathcal C^\perp
\stackrel{\scriptscriptstyle\mathrm{def}}=
\{ x\in \{0,1\}^\alpha \times \{0,1,2,3\}^\beta \mid [x,y]=0\mbox{ for all }y\in \mathcal C\}.
\end{equation}
Readily, $\mathcal C^\perp$ is also an additive code.
Moreover, $(\mathcal C^\perp)^\perp=\mathcal C$, see, e.g., \cite{BFPRV2010}.
\section{Symmetry Group}\label{s:sym}
Let $S_n$ be the set of permutations of $\{0,1,\ldots ,n-1\}$.
The \em symmetry group \em of a code $C\subset \{0,1\}^n$ is defined as
$$ \mathrm{Sym}(C) \stackrel{\scriptscriptstyle\mathrm{def}}= \{ \sigma \in S_n \mid \sigma(x)\in C \mbox{ for all } x\in C \}. $$
Given a $Z_2Z_4$-structure $\pi$, we will also consider the \em group of $Z_2Z_4$-symmetries \em $\mathrm{Sym}_{\pi}(C)$ as a subgroup
of $\mathrm{Sym}(C)$ consisting of symmetries $\sigma$ that
commute with the involution $\pi$:
$$\mathrm{Sym}_{\pi}(C) \stackrel{\scriptscriptstyle\mathrm{def}}= \left\{ \sigma \in \mathrm{Sym}(C) \mid 
\sigma(\pi(i))=\pi(\sigma(i)) \mbox{ for all } i\in \{0,1,\ldots ,n-1\} \right\}. $$
 In other words, $\mathrm{Sym}_{\pi}(C)$ is the intersection of $\mathrm{Sym}(C)$
with the automorphism group of the group $(\{0,1\}^n, *_{\pi})$ (which is, by definition, the set of
all permutations $\sigma$ of $\{0,1\}^n$ such that $\sigma(x) *_{\pi} \sigma(y) = \sigma(x *_{\pi} y) $ for every $x$, $y$).

The group $\mathrm{Sym}_{\pi}(C)$ 
has a natural treatment 
in terms of the preimage of $C$
under the Gray map.
Indeed, if 
$C = \Phi(\mathcal{C})$ 
for some  
$\mathcal{C}\subseteq \{0,1\}^\alpha \times \{0,1,2,3\}^\beta$ 
and $\pi$ is of form 
(\ref{eq:standard-pi}), 
then 
\begin{equation}\label{eq:sym-aum}
 \mathrm{Sym}_{\pi}(C) = 
 \{\Phi \sigma \Phi^{-1} \mid
 \sigma \in \mathrm{MAut}(\mathcal C)\},
\end{equation}
where $\mathrm{MAut}(\mathcal C)$, 
the \emph{monomial automorphism group} 
of $\mathcal C$,
is the stabilizer  
of $\mathcal C$ 
in the group 
of monomial transformations
of $\{0,1\}^\alpha \times \{0,1,2,3\}^\beta$
(recall that a \emph{monomial transformation}
consists of a coordinate permutation followed by sign changes in some quaternary coordinates).

To prove one of the corollaries from the main theorem, 
we will need the following simple known fact.
\begin{proposition}\label{p:ap-pa}
For every $Z_2Z_4$-additive code $C$, it holds 
$\mathrm{MAut}(\mathcal C) = \mathrm{MAut}(\mathcal C^\perp).$
\end{proposition}
\begin{proof}
At first, we see from (\ref{eq:inner}) that
$[\sigma^{-1}(x),y] = [x,\sigma(y)]$ 
for every monomial transformation $\sigma$.
This identity 
can be utilised 
to derive 
$\sigma(\mathcal C^\perp) = (\sigma(\mathcal C))^\perp$ from (\ref{eq:dual}):
\begin{eqnarray*}
\sigma(\mathcal C^\perp)
&=& \{ \sigma(x) \mid [x,y]=0\ \forall y\in \mathcal C\}
= \{ x' \mid [\sigma^{-1}(x'),y]=0 \ \forall y\in \mathcal C\}\\
&=&\{ x' \mid [x',\sigma(y)]=0 \ \forall y\in \mathcal C\}
=\{ x' \mid [x',y']=0 \ \forall y'\in \sigma(\mathcal C)\} = (\sigma(\mathcal C))^\perp.
\end{eqnarray*}
If $\sigma\in \mathrm{MAut}(\mathcal C)$,
then $\sigma(\mathcal C)=\mathcal C$
and hence $\sigma(\mathcal C^\perp) = \mathcal C^\perp$.
We conclude that
$\mathrm{MAut}(\mathcal C) \subseteq \mathrm{MAut}(\mathcal C^\perp)$.
From $(\mathcal C^\perp)^\perp=\mathcal C$, we get the inverse inclusion.
\end{proof}

It is worth to mention the \em full automorphism group \em $\mathrm{Aut}(C)$ of a binary code $C$, which is the stabilizer of the code in the group of isometries of the Hamming space.
We only observe that for a $Z_2Z_4$-linear code $C$, 
this group is the product of 
$\mathrm{Sym}(C)$ 
and the group of translations $\{\mathrm{tr}_c \mid c\in C\}$, 
where $\mathrm{tr}_c(x) \stackrel{\scriptscriptstyle\mathrm{def}}= c *_{\pi} x$. 
As follows, $|\mathrm{Aut}(C)|=|C|\cdot |\mathrm{Sym}(C)|$.
\section{Perfect and Extended Perfect Codes}\label{s:perf}

A binary code $C\subset \{0,1\}^n$ is called \em $1$-perfect \em(\em extended $1$-perfect\em )
if its cardinality
is $2^n/(n+1)$ (respectively, $2^{n-1}/n$) and the distance between
every two distinct codewords is at least $3$ (respectively, $4$),
where the (\emph{Hamming}) \emph{distance} is defined as the number of positions in which the words differ.
Note that the denominator $(n+1)$ (respectively, $n$)
coincides with the cardinality of a radius-$1$ ball (respectively, sphere) and must be a power of $2$
for the existence of corresponding codes.
So, a characterizing property of a $1$-perfect code is that every binary word
is at distance at most $1$ from exactly one codeword.
The codewords of an extended $1$-perfect code have the same parity
(i.e., the parity of the \emph{weight}, the number of ones in the word),
and every word of the other parity is at distance $1$ from exactly one codeword.

There is a characterization of $Z_2Z_4$-linear $1$-perfect and $Z_2Z_4$-linear extended
$1$-perfect binary codes, see  \cite{BorRif:1999} ($Z_2Z_4$-linear $1$-perfect codes), \cite{Kro:2000:Z4_Perf,Kro:arXiv:Z4_Perf} 
($Z_4$-linear extended $1$-perfect codes), and \cite{BorPheRif:2003} (complete description).
Recall that the \emph{rank} of a binary code is the dimension of its linear closure over $Z_2$.

\begin{proposition}[\cite{BorRif:1999,Kro:2000:Z4_Perf,BorPheRif:2003}]\label{p:char}
{\rm (a)} For any $r$ and $t \geq 4$ such that $ t/2 \leq r \leq t $,
there is exactly one $Z_2Z_4$-linear $1$-perfect code (extended $1$-perfect code) of type
$(2^r-1, 2^{t-1}-2^{r-1})$ (respectively, $(2^r, 2^{t-1}-2^{r-1})$), up to coordinate permutation.

{\rm (b)} For any $t \geq 4$,
there are exactly $\lfloor (t+1)/2 \rfloor$ $Z_2Z_4$-linear extended $1$-perfect codes of type
$(0, 2^{t-1})$ (i.e., $Z_4$-linear), up to coordinate permutation;
all these codes have different ranks $2^t - r - 1$, $r=\lfloor t/2 \rfloor,\ldots , t-1$,
except for the case $t=4$, $r=3$, where the corresponding code is linear.

{\rm (c)} All codes from (a) and (b) are pairwise nonequivalent. There are no other $Z_2Z_4$-linear $1$-perfect codes or $Z_2Z_4$-linear extended $1$-perfect codes.
\end{proposition}

\section{Preparata-Like Codes}\label{s:prep}
A binary code $C\subset \{0,1\}^n$ is called \em Pre\-pa\-ra\-ta-like \em
if its cardinality
is $2^n/n^2$ and the distance between
every two distinct codewords is at least $6$. 
Such codes exist if and only if $n$ is a power of $4$  
\cite{Preparata:1968}.
The original Preparata code \cite{Preparata:1968}, and the generalizations 
\cite{Dumer:1976}, \cite{BvLW}, \cite{vDamFDFla:2000} are not $Z_4$-linear if $n>16$.
A class of $Z_4$-linear Pre\-pa\-ra\-ta-like codes was constructed in \cite{HammonsOth:Z4_linearity}
for every $n=2^{t+1}\geq 16$, $t$ odd;
codes nonequivalent to that from \cite{HammonsOth:Z4_linearity}
were found in \cite{CCKS:97}.
As was shown in \cite[Theorem~5.11]{Kantor:2004symplecticsemifield}, 
there are many nonequivalent $Z_4$-linear Pre\-pa\-ra\-ta-like codes of the same length
(their number grows faster than any polynomial in $n$; 
however, there are some restrictions on $n=2^{t+1}$: 
$t$ is not a prime nor the product of two primes).
We will use the following two facts, which make 
our results concerning Pre\-pa\-ra\-ta-like codes simple corollaries 
from the results on extended $1$-perfect codes.

\begin{proposition}[\cite{ZZS:73:Praparata}]\label{p:P-C}
For every Pre\-pa\-ra\-ta-like code $P$, there exists a unique extended $1$-perfect code $C$ including $P$.
\end{proposition}

\begin{proposition}%
[\cite{BPRZ:2003:Preparata}]%
\label{p:ZP-ZC}
Assume that a Pre\-pa\-ra\-ta-like code $P$ 
is closed with respect 
to the operation $*_\pi$,
where $\pi$ is a $Z_2Z_4$ structure.
Then the extended $1$-perfect code $C$
including $P$ 
is also closed with respect 
to $*_\pi$.
\end{proposition}

In \cite{BPRZ:2003:Preparata}, 
it was shown that 
any $Z_2Z_4$-linear Pre\-pa\-ra\-ta-like code
is necessarily $Z_4$-linear,
i.e., the involution $\pi$ 
has no fixed points.

\begin{remark} 
In the current work, 
we consider the distance-$6$ Pre\-pa\-ra\-ta-like codes, 
sometimes referred to as 
the extended Pre\-pa\-ra\-ta-like codes.
In one-to-one correspondence 
with such codes are 
the distance-$5$ Pre\-pa\-ra\-ta-like codes, 
sometimes called 
the punctured Pre\-pa\-ra\-ta-like codes
(in fact, the original Preparata codes \cite{Preparata:1968} were presented in terms of distance-$5$ codes).
Reformulating Proposition~\ref{p:P-C},
every punctured Pre\-pa\-ra\-ta-like code
is included in a unique $1$-perfect code.
Formally, we can include 
the punctured Pre\-pa\-ra\-ta-like codes
in the statement of Theorem~\ref{th:main} below, 
but this does not make any sense 
as there are no
$Z_2Z_4$-linear codes among them, see \cite{BPRZ:2003:Preparata}.
\end{remark}

\section{Results}\label{s:res}

Generally, 
a binary code 
can admit 
more than one 
$Z_2Z_4$ structure. 
For example, 
the $1$-perfect code 
$\{0000000$, $
0001011$, $
0010110$, $
0101100$, $
1011000$, $
0110001$, $
1100010$, $
1000101$, $
1110100$, $
1101001$, $
1010011$, $
0100111$, $
1001110$, $
0011101$, $
0111010$, $
1111111\}$ (this is the cyclic Hamming code of length $7$, see e.g. \cite{MWS}) is closed with respect to $*_\pi$ for $22$ different involutions $\pi$,
including $\mathrm{Id}$,
$(01)(24)$, $(02)(14)$, and $(04)(12)$ (and all their cyclic shifts).
From the characterisation of 
$Z_2Z_4$-linear $1$-perfect codes,
we see that a $1$-perfect code of length at least $15$ or an extended $1$-perfect code of length more than $16$ cannot admit two
$Z_2Z_4$ structures with different number of self-adjacent coordinates.
The next theorem, which is the main result of the paper, states more.
\begin{theorem}\label{th:main}
Let $C$ be a $1$-perfect, extended $1$-perfect, or Pre\-pa\-ra\-ta-like code of length $n>16$
closed with respect to both operations $*_\pi$ and $*_{\tau}$, where $\pi$ and $\tau$ are $Z_2Z_4$ structures.
Then $\pi = \tau$.
\end{theorem}
The theorem will be proven in the next section. Here, we consider some important corollaries.
\begin{corollary}\label{cor:aut}
Let $C$ be a $1$-perfect, extended $1$-perfect, or Pre\-pa\-ra\-ta-like code of length $n>16$. 
If $C$ is $Z_2Z_4$-linear, 
i.e, closed with respect 
to the operation $*_\pi$,
for some $Z_2Z_4$ structure $\pi$,
then 
$\mathrm{Sym}(C)=\mathrm{Sym}_{\pi}(C)$.
\end{corollary}

\begin{proof}
 Seeking a contradiction, assume that $\sigma$ is from $\mathrm{Sym}(C)$ but not from $\mathrm{Sym}_{\pi}(C)$. 
 Then the involution $\tau \stackrel{\scriptscriptstyle\mathrm{def}}= \sigma^{-1}\pi\sigma$ does not coincide with $\pi$.
 However, for any two codewords $x$ and $y$,
 \begin{eqnarray*}
  x*_\tau y &=& x+y+(x+\tau(x))\cdot (y+\tau(y)) \\
  &=&\sigma^{-1} \left(x'+y'+\bigl(x' + \pi(x')\bigr)\cdot \bigl(y'+\pi(y')\bigr)\vphantom{A^1}\right) \\
 \end{eqnarray*}
 belongs to $C$ (here $x'=\sigma(x)\in C$ and $y'=\sigma(y)\in C$).
 We have a contradiction with Theorem~\ref{th:main}.
\end{proof}

An exhaustive computer search shows that the statement of Corollary~\ref{cor:aut} holds
also for the non-$Z_4$-linear extended $1$-perfect codes of lengths $16$ (as follows, it is also true
for the $Z_2Z_4$-linear $1$-perfect codes of length $15$, see the observation in the second paragraph of the next section).
For the $Z_4$-linear extended $1$-perfect codes of length $16$, the situation is different.
One of the two non-equivalent codes admits also the linear structure, and it is not difficult to find
that the $Z_4$-linear structure is not preserved by all symmetries.
The other code meets $|\mathrm{Sym}(C)| = 3|\mathrm{Sym}_{\pi}(C)|$.
The order of the symmetry group of the 
unique Pre\-pa\-ra\-ta-like code of length $16$ 
is $16\cdot 15\cdot 14\cdot12$ \cite{Berlekamp:71}, see also \cite{Bier:2007:NR}.


\begin{corollary}\label{c:abc}
{\rm (a)} 
If $C'$ is a $Z_2Z_4$-linear $1$-perfect code 
of type $(2^r-1, 2^{t-1}-2^{r-1})$, 
$t\ge 4$, 
$\frac t2\le r \le t$,
then
$\mathrm{Sym}(C')$ is isomorphic to the automorphism group of the group
$Z_2^{\dot\gamma}\times Z_4^{\delta}$, $\dot\gamma \stackrel{\scriptscriptstyle\mathrm{def}}= 2r-t$, $\delta \stackrel{\scriptscriptstyle\mathrm{def}}= t-r$,
and has the cardinality
$$2^{\frac12\dot\gamma^2-\frac12\dot\gamma+2\dot\gamma\delta+\frac32\delta^2-\frac12\delta}
\prod_{i=1}^{\dot\gamma} (2^i-1)\prod_{i=1}^\delta (2^i-1).$$

{\rm (b)} 
If $C$ is a $Z_2Z_4$-linear extended $1$-perfect code 
of type $(2^r, 2^{t-1}-2^{r-1})$, 
$t\ge 4$, $\frac t2\le r \le t$, then
$\mathrm{Sym}(C)$  is isomorphic to a semidirect product of
the automorphism group of the group $Z_2^{\dot\gamma}\times Z_4^{\delta}$, $\dot\gamma \stackrel{\scriptscriptstyle\mathrm{def}}= 2r-t$, $\delta \stackrel{\scriptscriptstyle\mathrm{def}}= t-r$,
 with the group $Z_2^{\dot\gamma+\delta}$ 
 of translations by an element of order less that $4$ 
 and has the cardinality
$$2^{\frac12\dot\gamma^2+\frac12\dot\gamma+2\dot\gamma\delta+\frac32\delta^2+\frac12\delta}
\prod_{i=1}^{\dot\gamma} (2^i-1)\prod_{i=1}^\delta (2^i-1).$$

{\rm (c)} 
If $C$ is a $Z_4$-linear extended $1$-perfect code 
of rank $2^t - r - 1$, $t>4$, 
$\frac{t-1}2\le r \le t-1$, then
$\mathrm{Sym}(C)$ 
has the cardinality
$$2^{\frac12\gamma^2+\frac32\gamma+2\gamma\dot\delta+\frac32\dot\delta^2+\frac52\dot\delta+1}
\prod_{i=1}^\gamma (2^i-1)\prod_{i=1}^{\dot\delta} (2^i-1),$$
where $\gamma \stackrel{\scriptscriptstyle\mathrm{def}}= 2r-t+1$, 
$\dot\delta \stackrel{\scriptscriptstyle\mathrm{def}}= t-r-1$.
\end{corollary}

\begin{proof}
(b,c)
Without loss of generality, 
we can assume that
the $Z_2Z_4$ structure corresponding 
to the considered $Z_2Z_4$-linear
code $C$ has the form (\ref{eq:standard-pi}),
where $\alpha = 2^r$ in the case (b) and $\alpha = 0$ in the case (c).

By Corollary~\ref{cor:aut}, 
$\mathrm{Sym}(C)
=\mathrm{Sym}_{\pi}(C)$ 
(the case $t=4$ 
is covered 
by the computational results 
mentioned above).
Without loss of generality assume that $\pi$ is of form 
(\ref{eq:standard-pi}).
Denote 
$\mathcal C \stackrel{\scriptscriptstyle\mathrm{def}}= \Phi^{-1}(C)$.
According to (\ref{eq:sym-aum}),
we have
$\mathrm{Sym}_{\pi}(C)\simeq
\mathrm{MAut}(\mathcal C)$.
By Proposition~\ref{p:ap-pa},
$\mathrm{MAut}(\mathcal C)=
\mathrm{MAut}(\mathcal C^\perp)$.
For a $Z_2Z_4$-linear extended $1$-perfect code $C$,
the related  code $\mathcal C^\perp$ is a so-called $Z_2Z_4$-additive Hadamard code, see \cite{PheRifVil:2006}.
The structure of the monomial automorphism group of such codes was studied in \cite{KroVil:2015},
and statements (b) and (c)
of the current corollary follow 
from \cite[Theorem~3]{KroVil:2015} 
and \cite[Theorem~2]{KroVil:2015}, respectively. 

(a) Since by appending a parity-check bit, every code $C'$ from p.(a) results in a code $C$ from p.(b),
where the new coordinate is self-adjacent,
$\mathrm{Sym}_{\pi}(C')$ coincides with the stabilizer of a self-adjacent coordinate in $\mathrm{Sym}_{\pi}(C)$.
As follows directly 
from the structure
of $\mathrm{MAut}(\mathcal C^\perp)$ considered in \cite[Section~IV]{KroVil:2015} and from the mentioned above connection between $\mathrm{MAut}(\mathcal C^\perp)$ and $\mathrm{Sym}_{\pi}(C)$,
the last group contains a subgroup isomorphic to $Z_2^r$ that acts transitively on the $2^r$ self-adjacent coordinates
(in the statement (b) of the current corollary, this subgroup is mentioned as the group of translations).
By the orbit--stabilizer theorem, we have 
$|\mathrm{Sym}_{\pi}(C')|=|\mathrm{Sym}_{\pi}(C)|/2^r$.
\end{proof}

Another interesting corollary from Theorem~\ref{th:main} was suggested by one of the reviewers. Recall that two
$Z_2Z_4$-linear or $Z_2Z_4$-additive codes are \emph{equivalent}
if one of the codes can be obtained from the other by a 
monomial transformation, that is, by a coordinate permutation and, if necessary, sign changes in some coordinates.
\begin{corollary}\label{c:rev}
Let $\mathcal C$ and $\mathcal D$ be 
$Z_2Z_4$-additive codes such that
$C=\Phi(\mathcal C)$ and
$D=\Phi(\mathcal D)$ are 
$1$-perfect, extended $1$-perfect,
or Pre\-pa\-ra\-ta-like codes.
If $\mathcal C$ and $\mathcal D$
are nonequivalent,
then $C$ and
$D$ 
are nonequivalent too.
\end{corollary}
\begin{proof}
Both $C$ and $D$
are closed with respect to $*_\pi$ where $\pi$ is
in the form (\ref{eq:standard-pi}).
 Assume that 
$C$ and
$D$ are equivalent; i.e., 
$C = \sigma(D)$
for some coordinate permutation 
$\sigma$.
Then, $C$ is also closed with respect to $*_{\sigma \pi \sigma^{-1}}$.
By Theorem~\ref{th:main}, we have $\pi = \sigma \pi \sigma^{-1}$.
This means that $\sigma$ preserves the pairs of adjacent coordinates.
It follows that $\Phi^{-1}\sigma\Phi$ is a monomial transformation 
and the codes $\mathcal D$ and 
$\mathcal C = \Phi^{-1}(C)=\Phi^{-1}(\sigma(D))=\Phi^{-1}(\sigma(\Phi(\mathcal D)))$ are equivalent.
\end{proof}
For the $Z_4$-linear Pre\-pa\-ra\-ta-like codes, this fact is new, 
as their class has not been completely characterized, 
in contrast to the case of (extended) $1$-perfect codes. 
It should be remarked, however, 
that the proof of Theorem~10.3(ii,iv) in \cite{CCKS:97} 
stating the same as Corollary~\ref{c:rev} 
for a partial class of $Z_4$-linear Pre\-pa\-ra\-ta-like codes 
works for all $Z_4$-linear Pre\-pa\-ra\-ta-like codes as well, 
taking into account the later result \cite{BPRZ:2003:Preparata} 
that only the $Z_4$-linear extended perfect code of length $2^t$ 
that has rank $2^t-t$ can include a $Z_4$-linear Preparata-like code.
\section{Proof of Theorem~\ref{th:main}}\label{s:proof}
\begin{proof}
We first note that by Propositions~\ref{p:P-C} and~\ref{p:ZP-ZC}, the statement on the Pre\-pa\-ra\-ta-like codes 
is straightforward from the one on the extended $1$-perfect codes.
Indeed, the only extended $1$-perfect code including a given $Z_2Z_4$-linear Pre\-pa\-ra\-ta-like code $P$
must be $Z_2Z_4$-linear with the same $Z_2Z_4$ structure as $P$.

At second, appending a parity-check bit to every codeword of a  $Z_2Z_4$-linear $1$-perfect code $C'$ of type $(\alpha,\beta)$
results in an extended $1$-perfect code $C$ of type $(\alpha+1,\beta)$.
Moreover any symmetry of $C'$ is naturally extended to a symmetry of $C$,
which fixes the last (appended) coordinate. 
So, to prove the theorem, it is sufficient to consider the case of an extended $1$-perfect code $C$.

Let $C$ be a $Z_2 Z_4$-linear code with two $Z_2Z_4$ structures, $\pi$ and $\tau$.
Seeking a contradiction, assume that for some coordinate $i$, we have $\pi(i)\ne \tau(i)$. 
Without loss of generality, $\pi(i)\ne i$.
We will say that two coordinates $j$ and $j'$ are \emph{independent} if $j'\not \in \{j,\pi(j),\tau(j)\}$
(equivalently, $j\not \in \{j',\pi(j'),\tau(j')\}$).

Suppose that $C$ has two codewords $v=(v_0,\ldots ,v_{n-1})$ and $u=(u_0,\linebreak[1]\ldots ,\linebreak[2]u_{n-1})$ such that every nonzero coordinate 
of $v$ is independent from every nonzero coordinate 
of $u$, with the only exception $v_i = u_{\pi(i)} = 1$.
Then $v*_\tau u$ coincides with $v + u$
(indeed, the situation in the middle bolded part of Table~\ref{eq:tables}(a) never occurs in this sum;
in fact, only the first row and the first column occur), while 
$v*_\pi u$ differs from $v+u$ in the coordinates $i$ and $\pi(i)$. 
Since both $v*_\tau u$ and $v*_\pi u$ must belong to $C$,
we have a contradiction with the code distance $4$.

It remains to find such two codewords $v$, $u$. 
We restrict the search by the weight-$4$ codewords.
Let  $i$, $i_2$, $i_3$, $i_4$ and $\pi(i)$, $j_2$, $j_3$, $j_4$ be the ones of $v$ and the ones of $u$, respectively.
It is easy to choose $i_2$, $i_3$, $i_4$ independent from $\pi(i)$ and to choose $j_2$ independent from $i$, $i_2$, $i_3$, and $i_4$.
For every choice of $j_3$, the fourth one $j_4$ of $u$ is defined uniquely. There are at least $n-3\cdot 4-1$ ways to choose 
$j_3$ independent from $i$, $i_2$, $i_3$, and $i_4$. In at least $n-3\cdot 4-1-3\cdot 4$ of them, the resulting $j_4$ is also 
independent from $i$, $i_2$, $i_3$, and $i_4$. Since $n-3\cdot 4-1-3\cdot 4\ge 32 - 25 >0$, the result follows. 
\end{proof} 

\section{Acknowledgement}
The work was supported by the Russian Foundation for Basic Research (grants 10-01-00424 and 13-01-00463).
The author would like to thank the referees for their valuable comments and suggestions.


\providecommand\href[2]{#2} \providecommand\url[1]{\href{#1}{#1}}
 \def\DOI#1{{\small {DOI}:
  \href{http://dx.doi.org/#1}{#1}}}\def\DOIURL#1#2{{\small{DOI}:
  \href{http://dx.doi.org/#2}{#1}}}

\end{document}